\documentclass{ecai} 
\usepackage{latexsym}
\usepackage{amssymb}
\usepackage{amsmath}
\usepackage{amsthm} 
\usepackage{booktabs}
\usepackage{enumitem}
\usepackage{graphicx}
\usepackage{color}
\usepackage{threeparttable}
\newtheorem{theorem}{Theorem}
\newtheorem{lemma}[theorem]{Lemma}

\newcommand{\BibTeX}{B\kern-.05em{\sc i\kern-.025em b}\kern-.08em\TeX}
\usepackage[justification=raggedright,singlelinecheck=false]{caption} 
\begin{document}

%%%%%%%%%%%%%%%%%%%%%%%%%%%%%%%%%%%%%%%%%%%%%%%%%%%%%%%%%%%%%%%%%%%%%%%%

\begin{frontmatter}

%%% Use this command to specify your submission number.
%%% In doubleblind mode, it will be printed on the first page.

\paperid{123} 

%%% Use this command to specify the title of your paper.

\title{Deep Learning and Explainable AI: New Pathways to Genetic Insights}

% 在作者信息下方添加脚注定义
\author[a]{\fnms{Chenyu}~\snm{Wang}\footnote[\dagger]{Contributed equally to this work.}}
\author[a]{\fnms{Chaoying}~\snm{Zuo}\footnotemark[\dagger]}
\author[a]{\fnms{Zihan}~\snm{Su}\footnotemark[\dagger]}
\author[b]{\fnms{Yuhang}~\snm{Xing}}
\author[a]{\fnms{Lu}~\snm{Li}}
\author[a]{\fnms{Maojun}~\snm{Wang}\footnote[*]{Corresponding authors: Maojun Wang (Email:mjwang@mail.hzau.edu.cn) and Zeyu Zhang (Email:zhangzeyu@mail.hzau.edu.cn).}}
\author[a]{\fnms{Zeyu}~\snm{Zhang}\footnotemark[*]}

\address[a]{National Key Laboratory of Crop Genetic Improvement, Hubei Hongshan Laboratory, Huazhong Agricultural University, 430070, Hubei, China}

\address[b]{Agricultural College, Shihezi University, 832002, Xinjiang, China}

%%% Use this combinations of commands to specify all authors of your 
%%% paper. Use \fnms{} and \snm{} to indicate everyone's first names 
%%% and surname. This will help the publisher with indexing the 
%%% proceedings. Please use a reasonable approximation in case your 
%%% name does not neatly split into "first names" and "surname".
%%% Specifying your ORCID digital identifier is optional. 
%%% Use the \thanks{} command to indicate one or more corresponding 
%%% authors and their email address(es). If so desired, you can specify
%%% author contributions using the \footnote{} command.

% \author[A]{\fnms{First}~\snm{Author}\orcid{....-....-....-....}\thanks{Corresponding Author. Email: somename@university.edu.}\footnote{Equal contribution.}}
% \author[B]{\fnms{Second}~\snm{Author}\orcid{....-....-....-....}\footnotemark}
% \author[B,C]{\fnms{Third}~\snm{Author}\orcid{....-....-....-....}} 

% \address[A]{Short Affiliation of First Author}
% \address[B]{Short Affiliation of Second Author and Third Author}
% \address[C]{Short Alternate Affiliation of Third Author}

%%% Use this environment to include an abstract of your paper.

\begin{abstract}
Deep learning-based AI models have been extensively applied in genomics, achieving remarkable success across diverse applications. As these models gain prominence, there exists an urgent need for interpretability methods to establish trustworthiness in model-driven decisions. For genetic researchers, interpretable insights derived from these models hold significant value in providing novel perspectives for understanding biological processes. Current interpretability analyses in genomics predominantly rely on intuition and experience rather than rigorous theoretical foundations. In this review, we systematically categorize interpretability methods into input-based and model-based approaches, while critically evaluating their limitations through concrete biological application scenarios. Furthermore, we establish theoretical underpinnings to elucidate the origins of these constraints through formal mathematical demonstrations, aiming to assist genetic researchers in better understanding and designing models in the future. Finally, we provide feasible suggestions for future research on interpretability in the field of genetics. 
\end{abstract}
\end{frontmatter}
%%%%%%%%%%%%%%%%%%%%%%%%%%%%%%%%%%%%%%%%%%%%%%%%%%%%%%%%%%%%%%%%%%%%%%%%

\section{Introduction}
Deep learning has demonstrated powerful modeling and analytical capabilities in 3D genomics and regulatory genomics, enabling efficient mining of complex regulatory patterns from high-dimensional genomic data \cite{koo2020deep}. In 3D genomics, deep learning models can predict spatial features of chromatin, including dynamic changes of topologically associating domains (TADs) \cite{zhou2022sequence}, chromatin compartments \cite{zhou2022sequence}, and chromatin loop \cite{lin2025unveiling}, as shown in Figure \ref{fig:1}(a), decoding genome function from sequence through structure \cite{fudenberg2020predicting}. These models can also decipher cell-type-specific chromosomal spatial organization patterns \cite{schuette2025chromogen}, quantify spatial interaction strengths between genomic loci, and uncover their associations with gene expression regulation.  In regulatory genomics, deep learning facilitates the integration of multi-omics data (e.g., Hi-C \cite{lieberman2009comprehensive}, ATAC-seq \cite{klemm2019chromatin}, HiChIP \cite{mumbach2016hichip}) to predict regulatory variants of DNA methylation \cite{zhou2025deep}, predict potential motif-motif interactions \cite{avsec2021effective, fu2025foundation}, efficiently identify cis-regulatory elements \cite{zhang2025trapt} and infer cooperative transcription factor binding networks \cite{fu2025foundation}, as presented in Figure \ref{fig:1}(b, c).

\begin{figure}[h]
\centering
\includegraphics[width=1.07\linewidth]{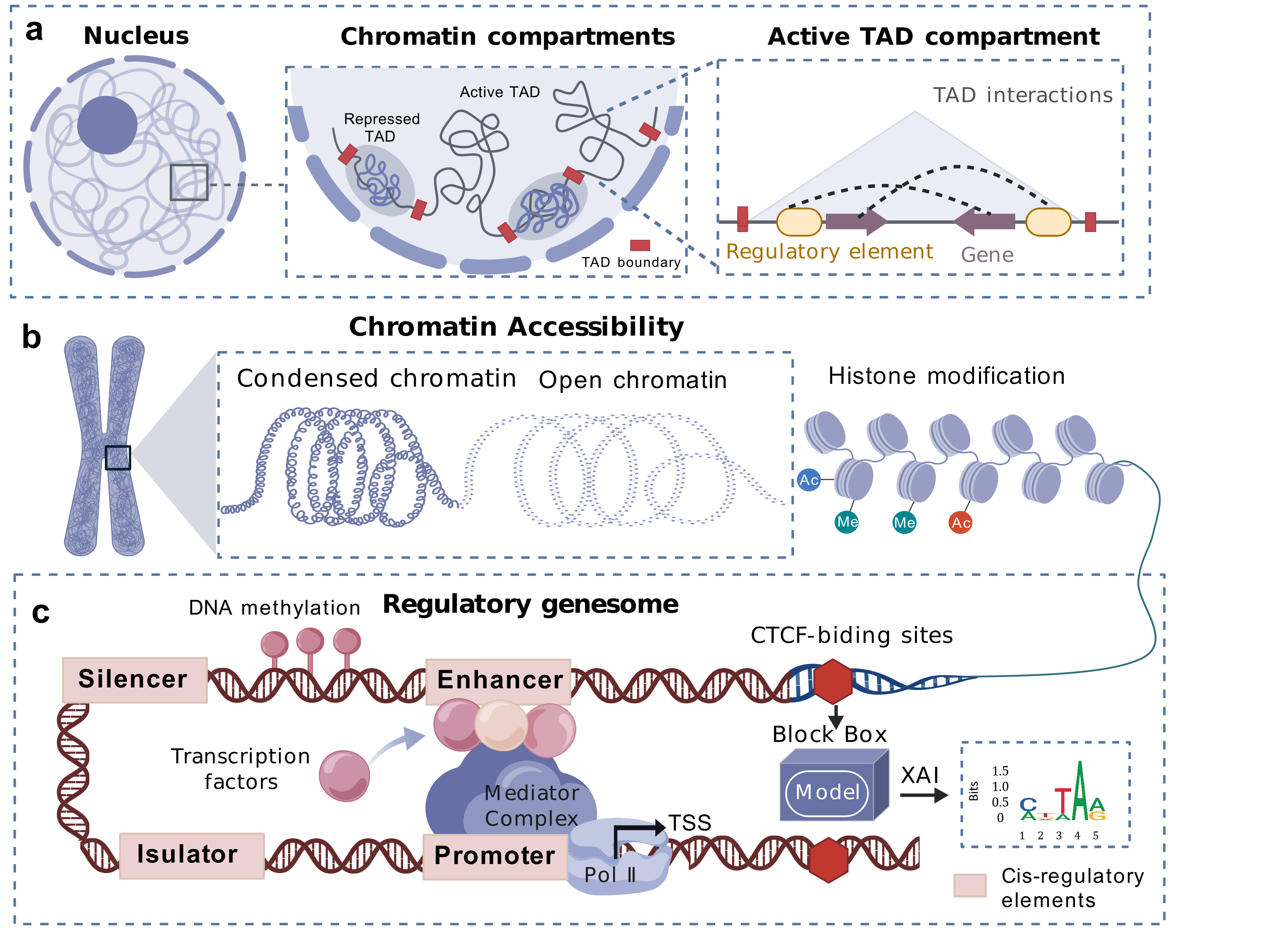}
\caption{\textbf{Applications of deep learning in 3D genomes and re- gulatory genomes.} \textbf{a.} Deep learning models can predict spatial structural features of  chromatin, such as chromatin compartments and topologically associating domains (TADs). \textbf{b.} Deep learning is also effective in predicting chromosome accessibility and histone modifications (e.g., H3K27ac). \textbf{c.} Additional applications include identifying cis-regulatory elements (e.g., enhancers/silencers), pre- dicting transcription factor binding motifs, and inferring precise transcription factor binding sites.} 
\label{fig:1}
\end{figure}

Although deep learning methods have achieved remarkable success in these biological fields, due to the use of multi-layer nonlinear transformations, complex architectures, and a large number of parameters, deep learning models are often perceived as `` black boxes'' by humans, making their internal mechanisms difficult to interpret and understand. The inherent `` black-box '' nature of these models not only limits their reliability and practical applicability,  but also makes it hard to understand how their predictions relate to actual biological processes, which hinders the translation of research findings into biological mechanisms \cite{azodi2020opening}.  

In recent years, the development of eXplainable AI (XAI) methods has emerged as a promising approach to address the `` black-box'' challenge \cite{molnar2020interpretable}. These methods aim to reveal the inner workings of neural network models in ways that are understandable to humans, or to explain how specific input features influence the prediction of model, thereby enhancing model transparency and trustworthiness. Therefore, many recent studies have employed interpretable deep learning methods to gain biologically relevant insights, such as identifying cis-regulatory elements that influence gene expression \cite{tan2023cell}. In this work, we systematically review and categorize the current interpretability techniques accordingly applied to two major tasks: 3D genome modeling and gene regulatory network prediction.

Despite various interpretable deep learning methods have shown promising results, each comes with its own limitations, and there remains a lack of rigorous mathematical proofs regarding the specific limitations of their applicability. In particular, many critical claims in existing studies (e.g.,`` enforcing transparency techniques may compromise model performance'' \cite{novakovsky2023obtaining}) are primarily based on empirical observations and lack solid theoretical foundations. To address this gap, we provide theoretical analyses of the limitations of certain interpretability methods.

The contributions of this paper are as follows:  

• We propose a novel classification framework that categorizes interpretable deep learning methods in genomics into two types: (1) Input interpretability, encompassing convolutional kernel visualization, gradient-based methods, and perturbation-based methods. (2) Model interpretability, including attention mechanisms and transparent models based on biological prior knowledge.

• We not only provide intuitive explanations of each method and introduce their applications but also rigorously reveal the inherent limitations of selected interpretable approaches through mathematical derivation, thus offering theoretical support for method evaluation and selection.

% • We propose several future directions to enhance the interpretability of deep learning models.

\section{Interpretable methods}

In this section, we systematically introduce interpretable approaches in deep learning applications for genomics, categorizing them into two main types: input interpretability and model interpretability. Input interpretability methods aim to explore key features learned from the input or to identify input features that have a significant impact on the predictions, employing techniques such as convolutional kernel visualization to identify learned sequence motifs, gradient-based methods to quantify feature importance, and perturbation-based analyses to assess the functional impact of input modifications. In contrast, model interpretability approaches, enhance transparency through architectural designs, including attention mechanisms that reveal the relationships between features and biologically inspired transparent models that explicitly align network components with known biological entities. The methods are systematically summarized in Table \ref{table}.  
The introduction of each method is as follows.
% Together, these strategies bridge the gap between complex model operations and actionable biological insights, enabling validation of learned features against domain knowledge and fostering trust in predictive outcomes.

\subsection{Input Interpretability}
% Input interpretability methods clarify decision-making mechanisms of deep neural networks in genomic predictions by analyzing how models process input data. 
Input interpretability aims to reveal key patterns within the input features or to assess the importance of each feature for the prediction of the model. 
Unlike traditional k-mer-based feature engineering methods that directly reveal feature importance, deep neural networks employ multi-layer nonlinear transformations for automated feature extraction, resulting in opaque decision processes. To address this, there exist many input interpretability methods, which we further categorize into three classes: convolutional kernel visualization, gradient-based methods, and perturbation methods. Figure \ref{fig:2} illustrates the underlying principles of each method.
%% 一段整体的描述

\textbf{Convolutional kernel visualization.}
%%具体介绍（原始版本，需修改）
With the advancement of deep learning technology, convolutional neural networks (CNNs) have been widely applied in genomic research, such as enhancer prediction \cite{min2017predicting}, transcription factor binding site prediction \cite{zheng2021deep}, and chromatin accessibility prediction \cite{kelley2016basset}. CNNs can automatically extract features from DNA sequences.  Specifically, CNNs slide multiple filters across the DNA sequence in the convolutional layers, obtaining activation values for each position relative to each filter. These activation values are then used as features and passed to fully connected layers for prediction. As shown in Figure \ref{fig:2}(a), the convolutional kernel weights in the first layer are first converted to a position frequency matrix (PFM), followed by log-scaling to produce a standard position weight matrix (PWM) \cite{alipanahi2015predicting, novakovsky2023obtaining}, which captures a short sequence motif. Therefore, visualizing these filters can help us understand which sequence patterns the model has extracted that are crucial for prediction. A potential issue here is the constraints in weight learning: unconstrained or inappropriate constraints during weight learning may lead to scaling problems, causing the importance of certain sequences to be disproportionately emphasized or underestimated. In most studies, a commonly used visualization strategy involves computing the activation values of multiple sequences for each trained filter, indicating the degree of match between each sequence and the filter. The sequences with the highest activation values are then selected to generate a corresponding PWM, representing the motif learned by the filter. Then we can obtain motifs learned by multiple filters and match them with existing motif databases to verify whether the model has captured biologically relevant and important features.
For example, DeepEnhancer \cite{min2017predicting} uses a CNN to predict enhancers from DNA sequences. Its first convolutional layer applies 128 filters of length 8 to extract features. By visualizing these filters, the model identifies motifs and matches them with the known transcription factor binding site database JASPAR \cite{mathelier2016jaspar}, successfully recognizing biologically relevant motifs. Furthermore, Basset \cite{kelley2016basset} recovered a large number of known DNA-binding protein motifs by analyzing the parameters of the 300 convolutional filters in the convolutional layer.

 % Convolution and visualization.
\begin{figure}[htbp]
\centering
\includegraphics[width=1.20\linewidth]{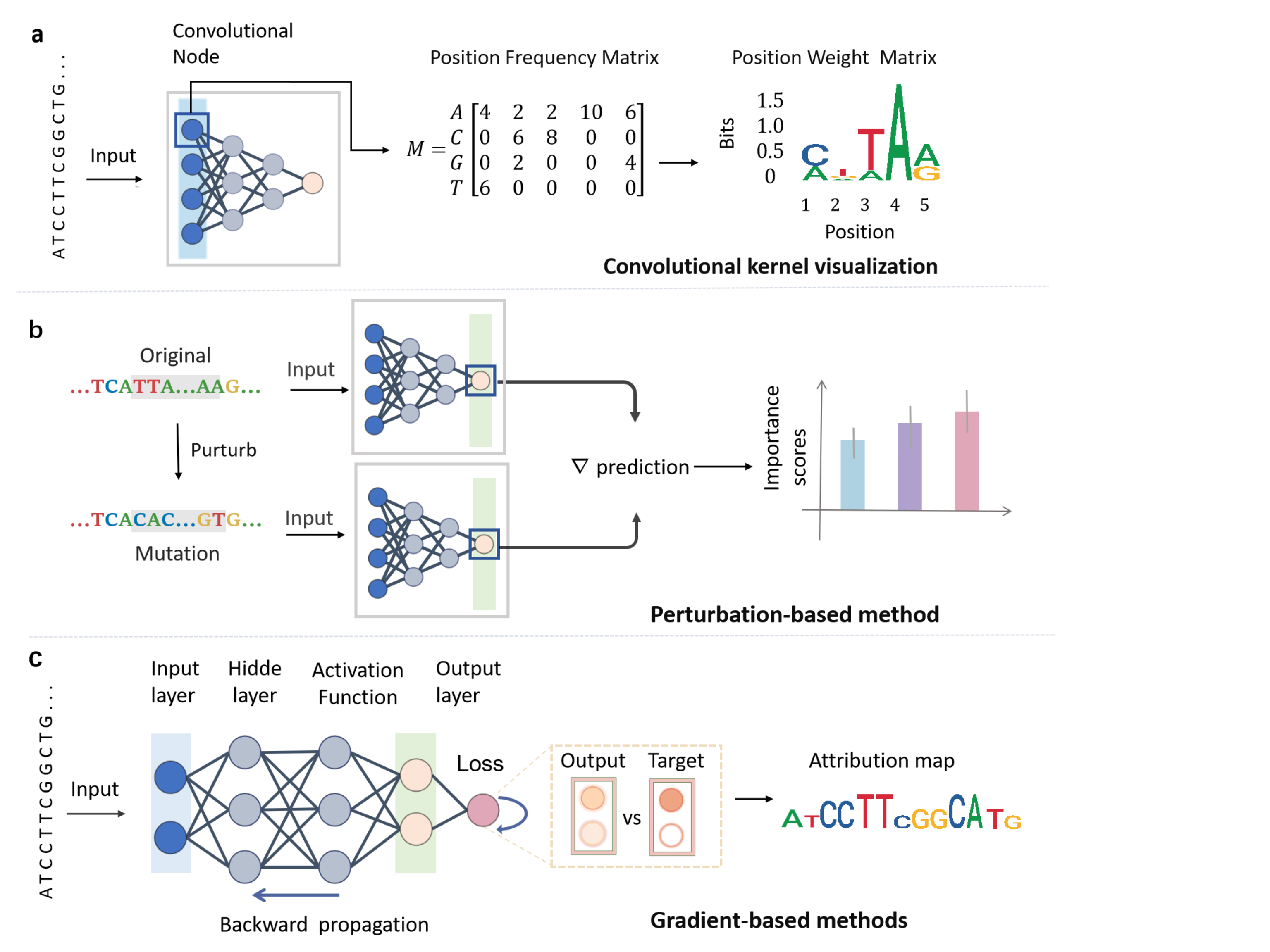}
\caption{Input Interpretability. \textbf{a.} Convolutional kernel visualiza- tion.  A single convolutional filter acts as a position weight matrix to scan the input sequence to quantify the binding bias of each position nucleotide in the DNA sequence. \textbf{b.} Perturbation-based methods. Systematic mutagenesis of individual nucleotides or sequence seg- ments enables quantitative evaluation of node-specific importance for model predictions. \textbf{c.} The gradient of the loss function computed by backpropagation highlights the contribution of each input feature.}
\label{fig:2}  
\end{figure}

% C|Gradient-based approach. 

\textbf{Gradient-based methods.}
%%具体介绍
Gradient-based importance analysis methods achieve global quantitative evaluation of feature importance by computing the partial derivatives of model outputs relative to input features. As shown in Figure \ref{fig:2}(c), this approach utilizes backpropagation to compute the gradients of the loss function with respect to sequence features, where gradient magnitudes directly measure the contribution of individual nucleotide in the DNA sequence. 
For instance, in genomic analyses, positive/negative gradient peaks correspond to enhancer and silencer regions, respectively \cite{kelley2018sequential}. 
However, gradient-based methods are susceptible to vanishing gradients or gradient saturation. In deep neural networks or when using saturating activation functions, excessively small gradients may introduce bias in feature importance estimation. This limitation can be mitigated by integrated gradients, which correct importance scores through path integration \cite{sundararajan2016gradients}. Moreover, an alternative approach, DeepLIFT \cite{shrikumar2017learning}, provides local explanations by comparing the prediction differences between test instances and reference sequences (e.g., background nucleotide frequencies), thereby avoiding the saturation effects of gradient methods. DeepLIFT has been employed to identify critical nucleotides in splice sites \cite{zuallaert2018splicerover}. The application of DeepLIFT algorithm successfully identified expression-predictive motifs (EPMs) in both 5'UTR and 3'UTR regions, which exhibit distinct sequence-specific positional preferences \cite{peleke2024deep}. Notably, the attribution rules of DeePLIFT match Shapley values, providing a robust feature importance framework \cite{lundberg2017unified}. Both integrated gradients and DeepLIFT require a predefined reference baseline, and the choice of this baseline can significantly influence the accuracy of attribution results. However, there is currently no consensus or established guideline on how to select an appropriate reference.

\begin{table*}[]
\caption{Interpretable Methods and Models Available for Deep Learning}
\label{table}
\begin{threeparttable}
\begin{tabular}{lllll}
\hline
\toprule%第一道横线
Tool Name and URL  & Year & Strategy & Model & Application \\
\midrule
\begin{tabular}[c]{@{}l@{}}Puffin  \cite{dudnyk2024sequence} \textit{(Science)} \\https://github.com/jzhoulab/puffin\end{tabular} & 2024 & 
Convolutional kernel visualization & U-net & Regulatory Genomics \\

\begin{tabular}[c]{@{}l@{}}INTERACT  \cite{zhou2025deep} \textit{(Science Advances)}\\https://zenodo.org/records/10955827\end{tabular} & 2025 & 
\begin{tabular}[c]{@{}l@{}}Convolutional kernel visualization\\ Attention Mechanism\end{tabular} & 
\begin{tabular}[c]{@{}l@{}}CNN, Transformer\end{tabular} & Regulatory Genomics \\

\begin{tabular}[c] {@{}l@{}}Basset \cite{kelley2016basset} \textit{(Genome research)}\\ http://www.github.com/davek44/Basset\end{tabular} & 2016 & 
Convolutional kernel visualization & CNN &Regulatory Genomics \\

\begin{tabular}[c]{@{}l@{}}EpiVerse \cite{lin2025unveiling} \textit{(Nature Communications)}\\https://github.com/jhhung/EpiVerse\end{tabular} & 2025 & 
\begin{tabular}[c]{@{}l@{}}Convolutional kernel visualization\\Perturbation-based methods\\ Gradient-based methods\end{tabular} & 
CNN, Transformer & 3D Genomics \\

\begin{tabular}[c]{@{}l@{}}Orca \cite{zhou2022sequence} \textit{(Nature genetics)}\\https://github.com/jzhoulab/orca\end{tabular} & 2022 & 
Perturbation-based methods & CNN & 3D Genomics \\

\begin{tabular}[c]{@{}l@{}}C.origama \cite{tan2023cell} \textit{(Nature biotechnology)}\\https://github.com/tanjimin/C.Origami\end{tabular} & 2023 & \begin{tabular}[c]{@{}l@{}}Perturbation-based methods\\Gradient-based methods\\Attention mechanism\end{tabular} & 
\begin{tabular}[c]{@{}l@{}}CNN, Transformer\end{tabular} & 3D Genomics \\

\begin{tabular}[c]{@{}l@{}}CLMLA \cite{dibaeinia2025interpretable} \textit{(Science Advances)}\\https://github.com/PayamDiba/CIMLA\end{tabular} & 2025 &Perturbation-based methods & Random Forest & Regulatory Genomics \\

\begin{tabular}[c]{@{}l@{}}Deepliver  \cite{bravo2024single} \textit{(Nature Cell Biology)}\\https://zenodo.org/record/8139953\end{tabular} & 2024 &\begin{tabular}[c]{@{}l@{}}Perturbation-based methods\\Gradient-based methods\end{tabular} & CNN & Regulatory Genomics \\

\begin{tabular}[c]{@{}l@{}}Geneformer \cite{theodoris2023transfer} \textit{(Nature)}\\https://huggingface.co/ctheodoris/Geneformer\end{tabular} & 2023 & 
\begin{tabular}[c]{@{}l@{}}Perturbation-based methods\\ Attention mechanism\end{tabular} & 
Transformer & Regulatory Genomics \\

\begin{tabular}[c]{@{}l@{}}DeepCRE \cite{peleke2024deep} \textit{(Nature Communications)}\\https://github.com/NAMlab/DeepCRE\end{tabular} & 2024 & \begin{tabular}[c]{@{}l@{}}Gradient-based methods\end{tabular} & 
\begin{tabular}[c]{@{}l@{}}CNN\end{tabular} &Regulatory Genomics \\

\begin{tabular}[c]{@{}l@{}}GET \cite{fu2025foundation} \textit{(Nature)}\\https://github.com/GET-Foundation\end{tabular} & 2025 & 
Gradient-based methods & Transformer & Regulatory Genomics \\

\begin{tabular}[c]{@{}l@{}}Enformer \cite{avsec2021effective} \textit{(Nature methods)}\\https://tfhub.dev/deepmind/enformer/1\end{tabular} & 2021 & 
\begin{tabular}[c]{@{}l@{}}Attention mechanism\\ Perturbation-based methods\\ Gradient-based methods\end{tabular} & 
Transformer & Regulatory Genomics \\

\begin{tabular}[c]{@{}l@{}}DCell \cite{ma2018using} \textit{(Nature methods)}\\http://d-cell.ucsd.edu/\end{tabular} & 2018 & 
\begin{tabular}[c]{@{}l@{}}Transparent models\\ Gradient-based methods\\ Perturbation-based methods\end{tabular} & 
Neural Network & 3D Genomics \\

\begin{tabular}[c]{@{}l@{}}Getnet \cite{van2021gennet} (\textit{Communications Biology)}\\https://github.com/arnovanhilten/GenNet/\end{tabular} & 2021 & 
Transparent models & Neural Network & Regulatory Genomics \\
\bottomrule
\end{tabular}
\end{threeparttable}

\end{table*}

\textbf{Perturbation-based methods.}
%%具体介绍
Perturbation-based methods  infer the importance of features by manipulating the input features and observing the changes in the output of the model. 
% In the field of computer vision, as deep neural networks have become the preferred approach for perception tasks such as image classification, their increasingly complex nature makes it difficult for people to understand how specific classifications are generated. To address this issue, perturbation-based methods were first applied int the field of computer vision. 
% By systematically occluding local regions of the input images\cite{zeiler2014visualizing}, the importance of features was judged based on the changes in the probability values of the model's prediction results. It is also possible to calculate the output differences to measure the importance by deleting or perturbing individual features\cite{zintgraf2017visualizing}. 
Perturbation-based interpretability methods were first introduced in the field of computer vision, where specific regions of an image are masked to observe changes in predictions, thereby assessing the importance of those regions \cite{zeiler2014visualizing}. This intuitive approach aligns well with human reasoning and has since been widely adopted in other domains.
Similarly to image perturbation, one can consider altering certain specific segments in biological gene sequences to determine the importance of the features corresponding to these segments. For example, we can perturb parts of the sequence and observe changes in the output, as seen in Figure \ref{fig:2}(b). The discrepancy between the predictions for these new sequences and the original sequence is quantified as the attribution score, serving as a critical metric for evaluating the functional significance of individual nucleotides in the context of the trained model. It is also feasible to mutate a specific nucleotide within the sequence into each of the other three nucleotides and observe the resultant changes in the output. By repeating this operation for all nucleotides, a matrix of dimensions $4\times L$ can be obtained, which is commonly referred to as the attribution map \cite{simonyan2013deep}. iMAP \cite{liu2022large} is used to knock out genes for the discovery of therapeutic targets for diseases. The perturbation model based on $\beta$-VAES \cite{bjerregaard2025interpretable} conducts perturbation analysis on single-cell RNA datasets to predict gene expression changes during gene knockout, toxin response, and embryonic development. 
 A perturbation map \cite{dhainaut2022spatial} was developed to provide a scalable approach for evaluating how specific genetic alterations impact the local, proximal, and distal tumor microenvironment (TME) states. However, many neural networks are trained in a way that resists Dropout \cite{hinton2012improving}. As a result, multiple neurons may be associated with the importance of the same feature, which can lead to the failure of perturbation methods. 

For input interpretability, while weight matrices corresponding to convolutional kernels can reflect the importance of sequences, the lack of proper constraints during their learning process may distort the perceived importance of these sequences. Gradient-based methods can evaluate feature importance by computing the partial derivatives of model output with respect to input features, while they may also face the potential issue of gradient vanishing. Perturbation-based methods assess the importance of altered features by introducing minor perturbations to the input and observing the resulting changes in model outputs, yet they may lead to an underestimation of feature importance due to neuronal redundancy.

\subsection{Model interpretability}

\begin{figure}[htbp]
\centering
\includegraphics[width=0.9\linewidth]{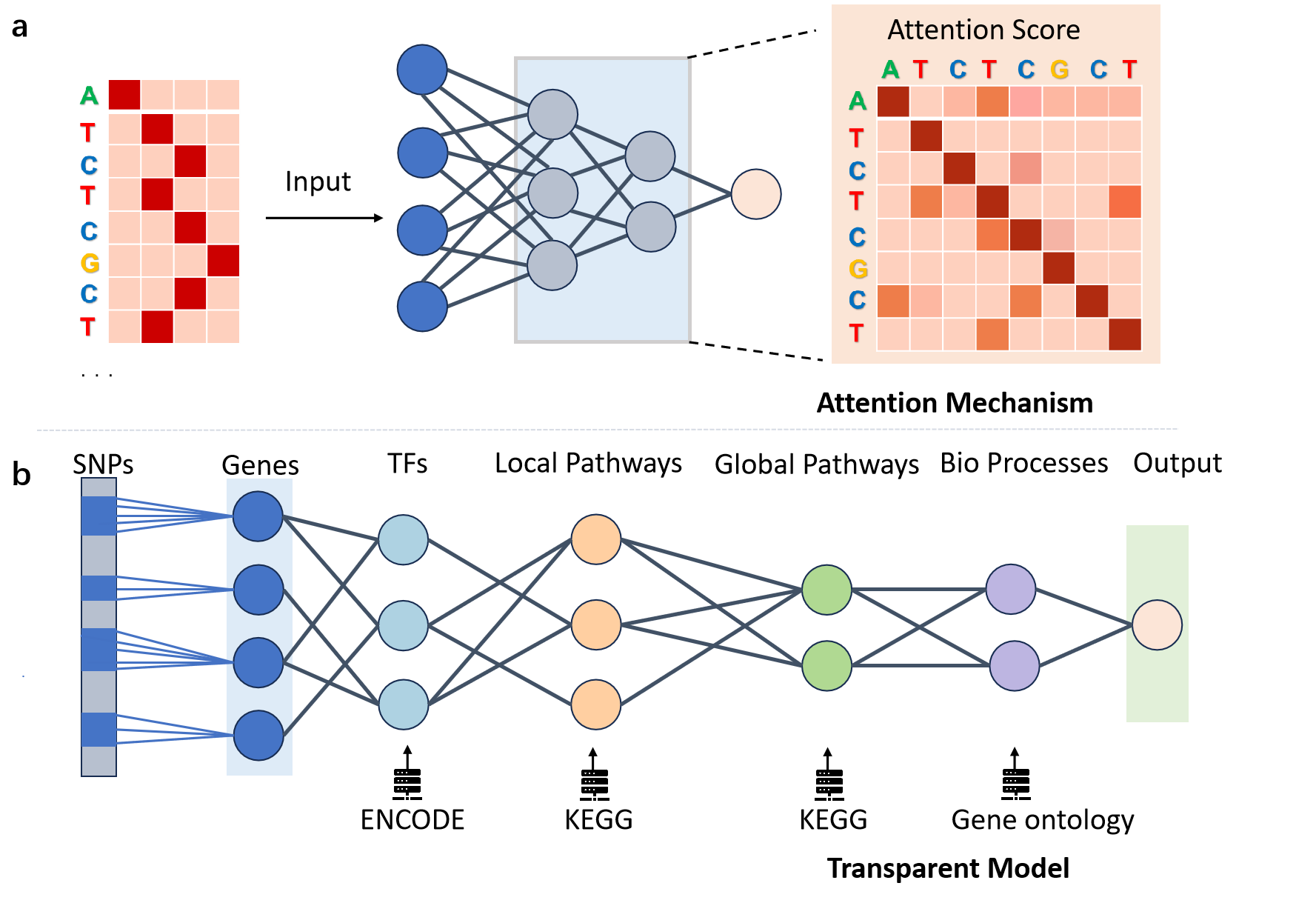}
\caption{\textbf{Model Interpretability.} \textbf{a.} Attention Mechanism. Quantify the interactions between different nucleotide positions in the sequence, thereby revealing combinatorial effects on model predictions or biological functions. \textbf{b.} Transparent models. The models are constructed by biological prior knowledge. }
\vspace{8pt}
\label{fig:Attn}
\end{figure}

%%一段整体的描述（原始版本，需修改）
% Model-based interpretable methods aim to help humans understand the model's decisions or provide interpretable information related to the model's predictions by designing the model's architecture or exploring specific components of the model.
Model interpretability aims to facilitate human understanding of internal mechanisms and decision-making processes through architectural design.
For example, traditional linear models are inherently interpretable, as we can directly understand the importance of each feature through its assigned weight. In contrast, deep neural networks typically consist of multiple hidden layers and numerous parameters, making it difficult to comprehend their decision-making processes solely based on their architecture or parameters. Therefore, many researchers are currently exploring model-based interpretability techniques to enhance the credibility of the model and uncover new biological insights. These methods can be mainly categorized into two types: attention mechanisms and transparent models based on biological prior knowledge. The principles of the various methods are illustrated in Figure \ref{fig:Attn}.

\textbf{Attention mechanisms.}
The attention mechanism is a widely adopted technique in neural networks. A representative example is the Transformer model \cite{vaswani2017attention}, which leverages self-attention to capture dependencies between tokens within a sequence. The attention weights inherently provide a form of interpretability \cite{naim2024explaining}, offering insight into which input tokens the model focuses on when making predictions and revealing interactions between them. The visualization results are shown in Figure \ref{fig:Attn}(a). EpiBERT \cite{javed2025multi} introduces a masking technique for genomic loci and enhancer regions, extracts query and key matrices from each attention layer, averages the resulting weights across all layers and attention heads, and visualizes these weights to indicate the degree of correlation between different positions. Enformer \cite{avsec2021effective} extracts the query row at the transcription start site (TSS), where the keys denote different spatial positions and the attention values reflect the focus of the model on these positions during the prediction of TSS. In addition, C.Origami \cite{tan2023cell} exploits attention scores from Transformer models to identify cis-regulatory elements that play critical roles in 3D chromatin architecture.
However, in DNA sequences, sequences at multiple positions may be associated with the same biological feature, which is mathematically termed multicollinearity, and this characteristic may lead to instability in the weights of attention mechanism matrices.

\textbf{Transparent models based on biological prior knowledge.}
%%具体介绍
Unlike traditional neural networks, where neurons in the hidden layer typically lack clear biological meaning, transparent neural network models are explicitly designed so that their hidden nodes correspond directly to specific biological units at a fine-grained level, thereby significantly enhancing interpretability \cite{ma2018using,elmarakeby2021biologically,van2021gennet}. To construct models with intrinsically interpretable units, it is necessary to incorporate prior knowledge into the network architecture design. For instance, by mapping bottom-layer nodes to molecular entities (e.g., genes or proteins) and deeper-layer nodes to functional modules (e.g., metabolic pathways or organelles), a hierarchical biological representation system can be established, as illustrated in Figure \ref{fig:Attn}(b). Deep neural networks build higher-order biological representations through progressive integration of lower-level features, for example, the second layer may capture cooperative interactions between transcription factor binding motifs, while deeper layers can map to complex systems such as complete biological pathways \cite{novakovsky2023obtaining}. As a breakthrough in biological applications, DCell \cite{ma2018using} pioneered a modeling approach that combines both neural network computational power and model transparency. The input layer of DCell directly corresponds to gene nodes, while its second layer constructs functional group nodes based on the hierarchical structure of Gene Ontology, establishing an association model between genotype and yeast growth rate. The GenNet \cite{van2021gennet} model utilizes genetic variants (SNPs) as input layer nodes and maps them to gene nodes in the second layer through NCBI RefSeq gene annotations, thereby constructing predictive association models between genotypes and complex phenotypes (e.g., schizophrenia, hair color). Through transparent variable construction based on multi-omics data, P-NET \cite{elmarakeby2021biologically} has identified MDM4 as a biomarker in metastatic prostate cancer. Although transparent models offer significant interpretability advantages, their construction relies on domain-specific prior knowledge, which to some extent limits their applicability. Furthermore, the technical approaches used to ensure model transparency may adversely affect predictive performance.  

For model interpretability, the attention mechanism reveals the importance of input tokens and reflects the correlations between different tokens. Multicollinearity among input tokens can lead to numerical instability in attention weights. Transparent neural networks can fully incorporate biologically relevant prior knowledge to reveal specific biological functional relationships, yet they may lack universality and cause models to overlook not clearly defined relationships.  

\section{Theoretical analysis}
In this section, we present a systematic theoretical analysis of the technical limitations inherent in each methodological approach.
\subsection{The limitation of input interpretability}
Analysis reveals three fundamental limitations in input interpretability approaches for deep neural networks: First, unconstrained weight optimization in convolutional layers provably causes scaling instability, mathematically demonstrating disproportionate feature amplification. Second, the inherent neuronal redundancy from dropout mechanisms obscures node-level importance attribution. Third, the null gradient of ReLU for negative inputs, combined with the multiplicative accumulation of small weights through successive layers, results in progressive gradient decay during backpropagation.

\textbf{The unconstrained learning of weights}. 
% In the convolutional neural network model for DNA sequence prediction, the first layer consists of convolutional nodes, which perform convolutional operations to obtain the feature sequences. In biological DNA problems, the scanned sequence is usually a sequence of size $L\times 4$. The number at a specific position in the sequence can be represented by $s_{ij}$, where $i$ represents the $i$-th position of the sequence, and $j$ $\in$ $\{1,2,3,4 \}$ represents the four bases of A, T, C, and G, respectively. And the convolutional kernel is usually a matrix of size $k\times 4$. By scanning the DNA sequence with the convolutional kernel, $L$ sequence feature values can be obtained.  
 CNNs for DNA sequence prediction extract local features by sliding multiple convolutional kernels across the input sequence. The learned kernels can ultimately be interpreted as PWM, representing the key motifs identified by the model from the input sequences. 
However, it should be noted that the learning of weights should be regulated by constraints. Otherwise, it is highly likely to lead to scaling issues. 

Suppose we study a classification problem where a DNA sequence is input to determine whether each position is related to a certain structure. The input sequences are typically represented as a $L\times 4$ matrix, where  $s_{ij}$ represents the entry at position (i, j) and  the convolutional kernel is a matrix of size $k\times 4$. To facilitate the analysis of the problem, the network structure is simplified. Specifically, the fully connected layers are abstracted into a single layer, and the binary cross-entropy loss function is adopted as the loss function. The sigmoid function $L=\sigma(z)$ is used as the activation function for the last layer of the neural network. The weights of the convolutional kernel are updated through  gradient descent. Then in each learning iteration, the weights will be updated in the following manner:
\begin{equation} \footnotesize
    w_{ij}\xleftarrow{}w_{ij}-\eta \frac{\partial L}{\partial w_{ij}}
\end{equation}

To further analyze the scaling problem of the weights, the partial derivatives will be expanded in accordance with the chain rule. Taking $w_{11}$, which is located at the first row and the first column of the convolutional kernel as an example:

\begin{equation} \footnotesize
    \frac{\partial L}{\partial w_{11}}=\frac{\partial L}{\partial z} \frac{\partial z}{\partial y}\frac{\partial y}{\partial w_{11}}
\end{equation}

{\footnotesize
\begin{align}  % 拼写错误：allign → align
    \frac{\partial L}{\partial z} &= -y \frac{1}{\sigma(z)}\sigma'(z)-(1-y)\frac{1}{1-\sigma(z)}(-\sigma'(z)) \\ \notag
    &= -y(1-\sigma(z))+(1-y)\sigma(z) \\ \notag
    &= \sigma(z)-y
\end{align}}

The sequence feature values $y_1,y_2,...,y_L$ are obtained by convolutional operation (using $y_1$ as an example):

\begin{align*} 
    y_{1}
    &=w_{11}s_{11}+w_{12}s_{12}+w_{13}s_{13}+w_{14}s_{14}\\ \notag
    &+w_{21}s_{21}+w_{22}s_{22}+w_{23}s_{23}+w_{24}s_{24}\\ \notag
    &+...\\ \notag
    &+w_{k1}s_{k1}+w_{k2}s_{k2}+w_{k3}s_{k3}+w_{k4}s_{k4} \\ \notag
\end{align*}

The output $z$ is obtained by multiplying the fully-connected weight matrix with the feature sequence.

\begin{equation}  \footnotesize
    z=\alpha_1 y_1+\alpha_2 y_2+...+\alpha_L y_L
\end{equation}

Furthermore, the update formula for the weights of the convolutional kernel can be specifically written.

\begin{equation}  \footnotesize
    \frac{\partial L}{\partial w_ {11}}=(\sigma(z)-y)(\alpha_1 s_{11}+\alpha_2 s_{21}+...+\alpha_L s_{L1})
\end{equation}

% \begin{equation}
%     w_{11}=w_{11}-\frac{\partial L}{\partial w_{11}}
% \end{equation}

Among them, $s_{ij}$ $\in$ $\{0,1 \}$. If the true label is 1, then $\sigma(z)-y<0$. If the subsequent part is greater than 0, it will cause the weight to increase. At this time, if there is no constraint, there will be a risk that the weight will continue to increase. Eventually, a certain weight will become too large. 

Then it is common practice to transform the PWM into a PFM\cite{novakovsky2023obtaining} through the softmax operation, and then obtain the standard PWM matrix via a logarithmic transformation.

The softmax operation on the convolutional kernel is shown in the following equation (still taking $w_{11}$ as an example):

\begin{equation}
     \footnotesize
w'_{11}=\frac{e^{w_{11}}}{e^{w_{11}}+e^{w_{12}}+e^{w_{13}}+e^{w_{14}}}
\end{equation}

It can be seen that if $w_{11}$ is too large while the other weight values are relatively small, it leads to $w'_{11} \approx \frac{e^{w_{11}}}{e^{w_{11}}+0+0+0}=1$.

And if $w_{11}$ is small and the other weight is relatively large, it will lead to $w'_{11} \approx \frac{0}{0+e^{w_{12}}+e^{w_{13}}+e^{w_{14}}}=0$.

This will lead to the irrationality of the weights. It will cause the functions of certain bases in the sequence to be overemphasized, while the significance of some other bases is neglected. 

For the purpose of ensuring the positive and negative nature of the values, when performing the logarithmic transformation, it is usually carried out by $w_{ij_{log}}=log(w'_{ij}+\epsilon)$, which $\epsilon$ is always determined according to the condition to satisfy the stability of the values. Due to $\epsilon $ being always small, when $w'_{ij}$ approaches 0, the value after logarithmic transformation will tend to negative infinity. When $w'_{ij}$ approaches 1, the value after the logarithmic transformation will tend to 0. Both scenarios will lead to the distortion of $w_{ijlog}$. 

This may lead to the effect of the bases in certain parts of the input DNA sequence being overly amplified, while the roles of the bases at other positions are ignored. This process theoretically analyzes why unconstrained learning can lead to scaling issues, thus rendering the interpretation of the DNA sequence lacking in reliability.

\textbf{The redundancy of neurons}.
Perturbation-based methods can evaluate the contribution of individual neurons by selectively removing them and observing the resulting changes in model predictions. However, deep neural networks are usually trained in a way that resists Dropout. 
The mechanism of Dropout is to deactivate some neurons during each training session. This ensures that certain specific features in the input DNA sequence do not rely on a fixed neuron, but can be captured by multiple neurons. 
More specifically, it is $r_{j}^{l} \sim Bernouli(p)$, which means that the $j$-th neuron in the $l$-th layer has the probability of $p$ being retained during training process. 
% \begin{align}
%     y^{\sim (l)}&=r^{(l)} y^{(l)}\\ \notag
%     z_{i}^{(l+1)}&=w_{i}^{(l+1)}y^{\sim (l)}+b_{i}^{(l+1)}\\ \notag
%     y_{i}^{(l+1)}&=f(z_{i}^{(l+1)})\\ \notag
% \end{align}

% Take the common parameter setting of $p=0.5$ as an example. In other words, assuming there are $n$ neurons in the entire neural network, the expected number of neurons retained in each training session is $\frac{n}{2}$. This implies that almost every time the training is carried out, it is essentially on a distinct neural network.

Suppose that there are two different neurons, $N_{i}$ and $N_{j}$, in the first convolutional layer, $y_{i}$ and $y_{j}$ are the final contributions of neuron $i$ and neuron $j$ in the entire neural network. 
\begin{align*} 
y_{i}&=f(\sum_{k}W_{ik}X_{k}+b_{i})\\ \notag
    y_{j}&=f(\sum_{k}W_{jk}X_{k}+b_{j}) \notag
\end{align*}

In standard CNNs, adjacent layer neurons are tightly coupled through weight updates, enabling collaborative feature learning. Dropout randomly disconnects the interneuron pathways during training, forcing independent feature acquisition. This causes redundant node formation: multiple subnetworks repeatedly learn identical DNA sequence features, meaning minimal output changes when removing specific neurons do not necessarily indicate low predictive importance of these nodes. 

% When the Dropout is not considered, the weights of the convolutional kernel of neuron $i$ may be updated under the influence of the output of neuron $j$. In this case, $Neuron_{i}$ and $Neuron_{j}$ is connected closely, and one neuron learn one feature and the other may learn another. We consider that neurons are interconnected with each other, jointly forming numerous pathways. For a neural network with a total of $k$ layers, where each layer has $n$ neurons, there are $n^{k}$ pathways. After the introduction of Dropout, it means that in some training sessions, the connection path between neuron $i$ and neuron $j$ is served. In this way, the two neurons learn features independently. Furthermore, each training session is equivalent to training a smaller sub-network of the entire neural network. Therefore, certain features i  the DNA sequence will be learned in each sub-network, which means they will be learned by many nodes in the entire network. This will lead to the redundancy of nodes. That is, if the output does not change significantly after a certain neuron, it does not directly indicate the importance of this node to the prediction result.

Since a neural network contains numerous neurons, as well as parameters such as weights and biases, it can be regarded as a high-dimensional manifold $\mathcal{M}$. The number of neurons in each neural network as a whole, the initial weights and other factors together constitute an abstract parameter $\theta$. From this perspective, the neural network training process is to select a minimum point of the loss function $L(\theta)$ on $\mathcal{M}$. Suppose that the overall parameter settings under the initial conditions are $\theta$. At this time, for any input $x$, the output of the model can be expressed as $f_{\theta}(x)$. $N_i$ and $N_j$ can be regarded as two distinct minor parameters within the parameter space $\theta$. Removing such a neuron is equivalent to perturbing in a certain parameter direction. If these two neurons are redundant with respect to each other, that is, there exist parameter directions $v_i,v_j \in \mathbb{R^d}$ such that for any input $x$:

\begin{equation}  \footnotesize
    f_{\theta+\epsilon v_i}(x)\approx f_{\theta+\epsilon v_j}(x)+O(\epsilon ^2)
\end{equation}

This equation indicates that the output results after small perturbations in two directions are approximately equal, where $O(\epsilon^2)$ denotes an extremely small error related to $\epsilon$. 

And these two parameter directions can be regarded as symmetric, and the corresponding covariant derivatives commute. 

\begin{equation}  \footnotesize
    \nabla v_i \nabla v_j f_\theta(x)=\nabla v_j \nabla v_i f_\theta (x)
\end{equation}

Then the components of the Riemann curvature tensor $\mathcal{R}$ of the parameter space in the directions of $v_1$ and $v_2$ are:

\begin{equation}  \footnotesize
     R(v_i,v_j,v_i,v_j)=<\nabla v_i\nabla v_jf_\theta - \nabla v_j \nabla v_i f_\theta ,f_\theta >\approx0
\end{equation}

It represents that the manifold is locally flat in the two parameter directions represented by the removal of redundant nodes. Assume that $N_i$ is removed. Then it can be regarded as a perturbation in the direction of $v_i$ in the parameter space, with $\theta'=\theta-\alpha v_i$. Subsequently, the variations in the output can be calculated: 
{\footnotesize
\begin{align*} 
    \Delta f&=f_{\theta}(x)-f_{\theta'}(x)\\ \notag
    &=f_{\theta}(x)-f_{\theta-\alpha v_1}(x)\\ \notag
    &=\nabla_{\theta}f_{\theta}(x)\cdot \alpha v_1-\frac{1}{2}(-\alpha v_1)^{T}\nabla_{\theta}^{2}f_{\theta}(x)(-\alpha v_1)+O(\alpha^3) \notag
\end{align*}}

Due to the large number and complexity of parameters in the entire neural network, the removal of a single neuron can be regarded as a minor change in a certain direction, and $\alpha$ is an infinitesimal quantity. At the same time, since the manifold $\mathbb{R}$ is locally flat in the subspace spanned by $v_1$ and $v_2$, $\nabla v_1 f_{\theta}(x)$ and $\nabla v_1^2f_{\theta}(x)$ approach 0. 

In summary, we obtain that $\Delta f$ approaches 0. That is, when redundancy occurs in neurons due to the Dropout mechanism, the output of the model does not change significantly after removing a single neuron, which makes it hard to access the importance of the node.

\textbf{Vanishing gradient in DNA sequences.}
 Vanishing gradient is a common issue in gradient-based model interpretation methods, 
 %This problem is also prominent when the input samples are DNA sequences, as these may contain multiple redundant features (e.g., a DNA sequence may contain several identical transcription factor binding sites). During backpropagation, when the gradients corresponding to these features progressively vanish, continued enhancement of these redundant features will exert diminishing influence on the model's output.
 and is closely tied to the choice of activation function. For instance, the ReLU activation function outputs zero for negative inputs while propagating positive inputs unchanged. However, its constant unity gradient for positive inputs may still lead to vanishing gradients during backpropagation in deep networks. Consider a single-layer neural network with the output given by:

\begin{equation}  \footnotesize
Y = \text{ReLU}(W \cdot X + b)\end{equation}

Here, \( W \) represents the weights matrix,  \( X \) is the input, and \( b \) is the bias term. The normal input and output relationship can be expressed as \( Y = \text{ReLU}(W \cdot X + b) = W \cdot X + b \). The gradient of the output \( Y \) with respect to the input \( X \) can be calculated as:

\begin{equation}  \footnotesize
\frac{\partial Y}{\partial X} = \frac{\partial \text{ReLU}(W \cdot X + b)}{\partial X} = \text{ReLU}'(W \cdot X + b) \cdot W
\end{equation}

According to the derivative definition of ReLU, if \( W \cdot X + b > 0 \), then \( \text{ReLU}'(W \cdot X + b) = 1 \), and the gradient is \( W \). Conversely, if \( W \cdot X + b \leq 0 \), then \( \text{ReLU}'(W \cdot X + b) = 0 \), and the gradient is 0. Extending this to a multi-layer neural network, assume output of each layer passes through a ReLU activation function, with the output of $l$-th layer given by:

\begin{equation}  \footnotesize
Z^{(l)} = \text{ReLU}(W^{(l)} \cdot Z^{(l-1)} + b^{(l)})
\end{equation}

Here, \( Z^{(0)} = X \) is the input, and \( Z^{(L)} = Y \) is the output. To calculate the gradient of the output \( Y \) with respect to the input \( X \), the chain rule must be applied:

\begin{equation}  \footnotesize
\frac{\partial Y}{\partial X} = \frac{\partial Z^{(L)}}{\partial Z^{(L-1)}} \cdot \frac{\partial Z^{(L-1)}}{\partial Z^{(L-2)}} \cdots \frac{\partial Z^{(1)}}{\partial X}
\end{equation}

The gradient of each layer is:

\begin{equation}  \footnotesize
\frac{\partial Z^{(l)}}{\partial Z^{(l-1)}} = \text{ReLU}'(W^{(l)} \cdot Z^{(l-1)} + b^{(l)}) \cdot W^{(l)}
\end{equation}

If the input to a layer \( W^{(l)} \cdot Z^{(l-1)} + b^{(l)} \leq 0 \), then \( \text{ReLU}'(W^{(l)} \cdot Z^{(l-1)} + b^{(l)}) = 0 \), and thus, the gradient for that layer is 0. This means that the gradient at this layer vanishes and cannot continue to propagate backward. Additionally, the gradient is also influenced by the weight matrices, as shown by:
{\footnotesize
\[  
\frac{\partial Y}{\partial X} = \frac{\partial Z^{(L)}}{\partial Z^{(L-1)}} \cdot \frac{\partial Z^{(L-1)}}{\partial Z^{(L-2)}} \cdots \frac{\partial Z^{(1)}}{\partial X} = W^{(L)} \cdot W^{(L-1)} \cdots W^{(1)}
\]}

When \( W^{(1)}, W^{(L-1)}, \ldots, W^{(L)} \) are very small, the value of the gradient will also be very small, leading to vanishing gradient.

% \textbf{Activation influences the interactions between features}
% In the field of bioinformatics, apart from focusing on the independent effects of each input feature, investigating the interaction effects between features is also of great significance. Integrated Hessians\cite{janizek2021explaining} is one of important method for determining interactions between features. 

% For input features $x_i$ and $x_j$, the interaction between them can be represented as

% \begin{equation}
%     T_{i,j}(x)=(x_i-x_i')(x_j-x_j')\int_0^1\int_0^1\alpha\beta\frac{\partial ^2 f(x'+\alpha \beta(x-x')}{\partial x_i\partial x_j}d\alpha d\beta
% \end{equation}

% $T_{i,j}(x)$ represents the degree of influence of the interaction between feature $i$ and feature $j$ on the model's output when the input is $x$. If the value of $T_{i,j}(x)$ is large, it indicates that feature $j$ has a significant impact on the extent to which feature $i$ affects the model's output. This implies that these two features are closely related during the model's decision - making process, and the role of feature $j$ cannot be ignored when analyzing the importance of feature $i$. 

% It can be seen from the equation that the magnitude of $T_{i,j}(x)$ depends on the second - order derivative. However, for the ReLU activation function, the second - order derivative of the function always holds that $\frac{\partial ReLU}{\partial x}=0$, then the feature interaction $T_{i,j}$ in the integrated Hessians will also be 0. It is unable to capture the real feature interactions. 

\subsection{The limitation of model interpretability}
This section systematically demonstrates the inherent limitations of model interpretability approaches in DNA sequence analysis. First, attention mechanisms intrinsically exhibit numerical instability, which can be mathematically proven to induce matrix ill-conditioning under multicollinearity. Second, biologically-constrained transparent models may yield higher loss values compared to their unconstrained counterparts. 

% Collectively, these findings provide theoretical evidence for the fundamental constraints of interpretable deep learning in genomics.  

%%具体介绍
\textbf{Instability of in the estimation of the weights.}
The computation of attention weights typically involves two vectors, namely the query vector and the key vector. 
% The query vector constitutes a mathematical representation of the focal target within the attention mechanism. Essentially, a key factor is a type of feature vector, which represents a form of encoding for the input information. 
The attention score is calculated by taking the dot product of the query matrix ${Q} =[q_1,q_2,...,q_n]$ and the key matrix ${K}=[k_1,k_2,...,k_n]$:
\begin{equation}  \footnotesize
    S=QK^{T}
\end{equation}

Since the input features often exhibit multicollinearity, some of the vectors in the key matrix are linearly correlated. That is, there exists an index $i$ such that $k_i=\sum_{j\neq i}\alpha_jk_j+\beta$.

Consider the covariance matrix $C$ of $K$, and $C=K^{T}K$, due to the multicollinearity, the rank of the matrix $K$ will satisfy $r(K)=r<n$, then we have:
\begin{equation}  \footnotesize
    r(C)=r(K^{T}K)\leq r(K) < n
\end{equation}

The matrix $C$ also satisfies the condition of partial linear correlation. Therefore, $|C|=0$. Since the matrix $C$ is a positive semi-definite matrix, all the eigenvalues of $C$ are greater than or equal to 0. At this time, there will be:
\begin{equation} \footnotesize
    |C|=\lambda_1\lambda_2\cdot\cdot\cdot \lambda_n=0
\end{equation}

So there must exist $\lambda_i$ which satisfies that $\lambda_i=0$, i.e., $\lambda_{min}=0$. 

To demonstrate that multicollinearity leads to numerical instability, we first present the following lemma and theorem. 

\begin{lemma}
    In the sense of the 2-norm, for any matrix $A$, if $\lambda_{max}$ is the largest eigenvalue of $A^TA$, then $||A||_{2}=\sqrt{\lambda_{max}}$.
\end{lemma}
\begin{proof}

   Arbitrarily select a matrix $A$ and any unit vector $x$ that can be transformed by $A$. According to the definition of the norm and the unit vector, we have:

   \begin{equation}  \footnotesize
       ||A||_{2}=sup_{||x||\neq 0}\frac{||Ax||}{||x||}=sup_{||x||\neq 0}||Ax||
   \end{equation}

Then write $||Ax||$ in another form:

   \begin{equation}  \footnotesize
       ||Ax||=\sqrt{(Ax)^T(Ax)}=\sqrt{x^TA^TAx}
   \end{equation}

Due to $x^TA^TAx\geq 0$, $A^TA$ is a positive semi-definite matrix. Suppose that its $n$ eigenvalues satisfy $\lambda_1 \geq \lambda_2 \geq ... \geq \lambda_n \geq 0$, the corresponding eigenvectors $\alpha_1, \alpha_2,..., \alpha_n$ form an orthonormal basis of $\mathbb{R}^n$. Let $x=\sum_{i=1}^{n}k_i \alpha_i$, since $||x||=1$, so $\sum_{i=1}^{n} k_i^2=1$. Subsequently, we can derive the following equation:

  \begin{equation}  \footnotesize
       A^TAx=A^TA(\sum_{i=1}^n k_i \alpha_i)=\sum_{i=1}^nk_i A^T A \alpha_i =\sum_{i=1}^n\lambda_i k_i \alpha_i
  \end{equation}

   \begin{equation}  \footnotesize
       x^TA^TAx=\sum_{i=1}^{n}\lambda_i k_i^2 \leq \lambda_1
   \end{equation}

   $\lambda_1$ is the largest eigenvalue of $A^TA$, which is marked as $\lambda_{max}$. Then, $\sqrt{(Ax)^T(Ax)} \leq \sqrt{\lambda_{max}}$. 

   Up to this point, we have obtained $||A||_{2}=\sqrt{\lambda_{max}}$.

\end{proof}

\begin{theorem}
    Multicolinearity leads to the condition number of the matrix tending to infinity.
\end{theorem}

\begin{proof}
    According to Lemma 1, we have $||A||_{2}=\sqrt{\lambda_{max}}$ and $||A^{-1}||_{2}=\frac{1}{\sqrt{\lambda_{min}}}$.

    \begin{equation}  \footnotesize
        k(A)=||A||_{2}\cdot ||A^{-1}||_{2}=\frac{\sqrt{\lambda_{max}}}{\sqrt{\lambda_{min}}}\rightarrow \infty
    \end{equation}

\end{proof}

This means that multicollinearity will lead to an excessively large condition number of the matrix, which in turn poses a potential risk of numerical instability \cite{zhang2024geometric}.

\textbf{Constructing transparent models with prior knowledge.}
% Currently, transparent models have garnered attention due to their enhanced interpretability and computational efficiency, which are facilitated by the incorporation of prior knowledge. These advantages position transparent models as having significant potential applications in biological research, particularly in scenarios that necessitate a deep understanding of biological mechanisms and the design of experimental validations.
Transparent models constructed based on prior knowledge have attracted widespread attention due to their computational efficiency and interpretability.
However, transparent models also have their limitations.  Biological constraints within transparent models may diminish their predictive performance. For example, Hard-coding TF-TF interactions may underperform unconstrained models. (although the use of milder regularization techniques can mitigate this decrease to some extent). The following serves as an illustration of this. Suppose that the DNA samples input into two models are identical, with all other confounding factors excluded. The loss function for the hard-coded transparent model is defined as follows:
\begin{equation} \footnotesize
L_{\text{hard}} = -\sum_{i=1}^{N} m_i \log(p_i^{\text{hard}}) + \lambda \| W_{\text{prior}} - W_{\text{ref\_prior}} \|_2 + \gamma \| W_{\text{adapt}} \|_2
\end{equation}

There, \( N \) denotes the number of input samples, \( m_i \) is the true label of the \( i \)-th sample, and \( p_i^{\text{hard}} \) represents the predicted probability of the \( i \)-th sample by the hard-coded model. This prediction is subject to the constraint \( p_i^{\text{hard}} = f(W_{\text{prior}}, W_{\text{adapt}}, m_i) \). The regularization coefficient \( \lambda \) is associated with the prior constraint term, while \( \gamma \) corresponds to the regularization coefficient for the adjustable parameters in the hard-coded model, which usually has a smaller value. \( W_{\text{prior}} \) signifies the fixed weights in the hard-coded model, \( W_{\text{ref\_prior}} \) is the preset baseline weight reference value, and \( W_{\text{adapt}} \) is the weight in the hard-coded model that can be freely adjusted. The loss function of the hard-coded model is primarily determined by the first term, which is the cross-entropy loss.
The loss function for the unconstrained model is defined as follows:

\begin{equation}  \footnotesize
L_{\text{free}} = -\sum_{i=1}^{N} m_i \log(p_i^{\text{free}}) + \beta \| W_{\text{free}} \|_2
\end{equation}

In comparison, the predicted probability \( p_i^{\text{free}} = f(W_{\text{free}}, m_i) \) of the unconstrained model is determined by the free parameters \( W_{\text{free}} \). Here, \( \beta \) is the regularization coefficient for the unconstrained model (with a smaller value). Similarly, the loss function of the unconstrained model is primarily determined by the first term, which is the cross-entropy loss. Due to the limitations imposed by hard coding, it holds that \( p_i^{\text{hard}} \leq p_i^{\text{free}} \).
To demonstrate that the performance of the hard-coded model may potentially decrease compared to the unconstrained model, it is only necessary to prove that the difference between the loss functions of the hard-coded model and the unconstrained model is not less than zero.
\begin{align*} 
L_{\text{hard}} - L_{\text{free}} &= \sum_{i=1}^{N} m_i \left( \log \frac{p_i^{\text{free}}}{p_i^{\text{hard}}} \right) + \lambda \| W_{\text{prior}} - W_{\text{ref\_prior}} \|_2 \\
&\quad + \gamma \| W_{\text{adapt}} \|_2 - \beta \| W_{\text{free}} \|_2 
\end{align*}

Where \( \sum_{i=1}^{N} m_i \left( \log \frac{p_i^{\text{free}}}{p_i^{\text{hard}}} \right) \geq 0 \), \( \lambda \| W_{\text{prior}} - W_{\text{ref\_prior}} \|_2 \geq 0 \), and \( \gamma \| W_{\text{adapt}} \|_2 - \beta \| W_{\text{free}} \|_2 \approx 0 \). From this, it can be concluded that \( L_{\text{hard}} - L_{\text{free}} \geq 0 \), which proves that the hard-coded model may perform worse than the unconstrained model.

\section{Future Directions}

Although we have theoretically analyzed the problems existing in various interpretable methods, how to optimize and improve these methods with limitations remains a direction worthy of exploration. 

Overall, on the one hand, research demonstrates that existing interpretation approaches, including attention-based visualization techniques, feature importance analysis, and causal reasoning-exhibit significant variations in terms of information presentation formats, analytical granularity, and reliability. Such heterogeneity not only poses substantial challenges for researchers in selecting and applying appropriate methods, but may also potentially compromise the credibility of model-derived decisions to some extent. Consequently, establishing a systematic evaluation framework for explainable methods becomes particularly crucial. Future research should focus on developing comprehensive and objective assessment metrics to construct a unified benchmarking standard, thereby effectively enhancing the reliability and practical utility of explainability studies.

On the other hand, we expect that interpretable information will not rely solely on a single method or a certain input information. With the development of AI and the improvement of computing power, multimodal technology is thriving. In the future, multimodal technology can be applied to integrate various forms of input data, such as gene sequence and protein structures. At the same time, by comprehensively considering the information provided by different types of interpretation methods, the understanding and reasoning abilities of the model can be enhanced, providing a more reliable guarantee for genetic predictions. 

\section{Conclusions}
With the extensive application of deep learning in the field of bioinformatics, reliable interpretability has become increasingly crucial for understanding the prediction results.  
This article provides a systematic overview of interpretability methods for deep learning applied in the field of bioinformatics and then categorizes these methods into two primary types: input interpretability and model interpretability, further detailing various interpretive approaches within each category. Additionally, we theoretically analyze the limitations of these interpretive methods and provide guidance for future research on interpretability in the field of genetics.

%%%%%%%%%%%%%%%%%%%%%%%%%%%%%%%%%%%%%%%%%%%%%%%%%%%%%%%%%%%%%%%%%%%%%%%%

%%%%%%%%%%%%%%%%%%%%%%%%%%%%%%%%%%%%%%%%%%%%%%%%%%%%%%%%%%%%%%%%%%%%%%%%

%%% Use this environment to include acknowledgements (optional).
%%% This will be omitted in doubleblind mode.

% \begin{ack}
% By using the \texttt{ack} environment to insert your (optional) 
% acknowledgements, you can ensure that the text is suppressed whenever 
% you use the \texttt{doubleblind} option. In the final version, 
% acknowledgements may be included on the extra page intended for references.
% \end{ack}

%%%%%%%%%%%%%%%%%%%%%%%%%%%%%%%%%%%%%%%%%%%%%%%%%%%%%%%%%%%%%%%%%%%%%%%%

%%% Use this command to include your bibliography file.

\bibliography{mybibfile}

\end{document}